\documentclass[12pt,final]{article}
\usepackage{amssymb,amsmath,amsthm}
\usepackage{epsfig}
\usepackage{graphics}
\newtheorem{theorem}{Theorem}[section]
\newtheorem{lemma}[theorem]{Lemma}

\newtheorem{definition}[theorem]{Definition}

\newtheorem{corollary}[theorem]{Corollary}
\newtheorem{example}[theorem]{Example}

\begin{document}
\title{Fundamental Constraints on Multicast Capacity Regions}
\author{Leonard Grokop \quad David N. C. Tse\\
Department of Electrical Engineering and Computer Sciences\\
University of California\\
Berkeley, CA 94720, USA\\
\{lgrokop,dtse\}@eecs.berkeley.edu}
\maketitle

\begin{abstract}
Much of the existing work on the broadcast channel focuses only on the sending of private messages. In this work we
examine the scenario where the sender also wishes to transmit common messages to subsets of receivers. For an $L$-user
broadcast channel there are $2^L-1$ subsets of receivers and correspondingly $2^L-1$ independent messages. The set of
achievable rates for this channel is a $2^L-1$-dimensional region. There are fundamental constraints on the geometry of
this region. For example, observe that if the transmitter is able to simultaneously send $L$ rate-one private messages,
errorfree to all receivers, then by sending the same information in each message, it must be able to send a single
rate-one common message, errorfree to all receivers. This swapping of private and common messages illustrates that for
any broadcast channel, the inclusion of a point ${\bf R^*}$ in the achievable rate region implies the achievability of
a set of other points that are not merely componentwise less than ${\bf R^*}$. We formerly define this set and
characterize it for $L=2$ and $L=3$. Whereas for $L=2$ all the points in the set arise only from operations relating to
swapping private and common messages, for $L=3$ a form of network coding is required.
\end{abstract}

\section{Introduction}

\noindent The broadcast channel has predominantly been studied in the context of unicast messaging, where the
transmitter wishes to send one private message to each of the $L$ receivers (see \cite{cover} for example). We refer to
this as unicasting. The transmitter may however wish to send different messages to different subsets of receivers. We
refer to this as multicasting. The most general multicast structure comprises of $2^L-1$ messages (the powerset). For
$L=2$ there are three messages, one required only by the first receiver, one required only by the second receiver, and
one required by both receivers.

The multicast capacity region for a broadcast channel is the set of $2^L-1$-dimensional rate vectors that are
achievable. For $L=2$ this is the set of achievable rate vectors $(R_{1},R_{2},R_{12})$, where $R_{12}$ denotes the
rate of the common message. One question of interest is, can the multicast capacity region be inferred from the unicast
capacity region? That is, can we always compute the multicast capacity region from the unicast capacity region, i.e.
without knowing the structure of the channel? For certain broadcast channels this is true, although it is not true in
general. Thus the multicast capacity region provides additional information about the communication limits of the
channel beyond that of the unicast capacity region.

Multicasting has received significant attention in the network-coding literature. In \cite{NetworkCoding} and \cite{KM}
the maximum rate at which a common message can be sent from a source node through a network of directed noiseless links
to a collection of sink nodes, is shown to equal the minimum-cut of the associated graph. In \cite{ErezFeder} and
\cite{Yeung} the multicast capacity region for one-source-two-sink networks is fully characterized.\footnote {To be
more precise, we define the multicast capacity region of a network as the convex-hull of the union of all multicast
capacity regions of broadcast channels that arise from specifying the encoding and decoding operations at intermediate
nodes in the network.} It is again given by the minimum-cuts of the associated graph. For three or more sinks this is
not the case and the problem is open. In this exposition we shed light on it by characterizing some of the structure
for three-sink networks.

There is an oddity to multicasting. Suppose we have a two-user broadcast channel that can support a rate vector
$(1,1,1)$. That is the transmitter can simultaneously deliver one bit of private information to the first receiver, one
bit of private information to the second receiver, and one bit of common information to both receivers. An important
point to clarify is that there is {\it no secrecy requirement} --``private'' information sent to the first receiver may
or may not be decodable by the second receiver and vice versa. Then the channel can also support a rate vector
$(2,1,0)$. The transmitter merely uses the common bit to send private information to the first receiver. Ofcourse the
second receiver is capable of decoding this bit too, but the information is of no interest to it. By symmetry the
achievability of rate vector $(1,1,1)$ also implies the achievability of rate vector $(1,2,0)$. There is one more
implication in this vein: the achievability of $(1,1,1)$ implies the achievability of $(0,0,2)$. The reasoning is
similar. The transmitter sends the same information on the two private bits. In this way the first user receives the
same private bit as the second user, in addition to the same common bit. Thus two common bits have been sent. These
three manipulations are summarized in figure \ref{fig:BC_example} as {\it extremal rays} stemming from $(1,1,1)$ and
represent three distinct encoding/decoding operations that can always be performed, regardless of the structure of the
broadcast channel. In this sense they are universal. By time-sharing one can achieve any point in the polytope
indicated in figure \ref{fig:BC_example}. To summarize: if a rate-vector $(1,1,1)$ is achievable, so must be the region
illustrated, regardless of the channel. Is this set of operations complete? Put in reverse, are there any rate vectors
outside the polytope in figure \ref{fig:BC_example} that are achievable on for all broadcast channels for which
$(1,1,1)$ is achievable? The answer is that there are not --there exists a broadcast channel where the rate vector
$(1,1,1)$ is achievable, but no rate vector outside the polytope in figure \ref{fig:BC_example} is. Thus for the
two-user broadcast channels the three operations discussed form a complete set -they are the only {\it distinct
universal encoding/decoding operations}.

It is straightforward to generalize these operations to broadcast channels with an arbitrary number of users. Consider
for example the three user broadcast channel. There are seven messages. Suppose a rate vector
$(R_1,R_2,R_3,R_{12},R_{13},R_{23},R_{123}) = (1,1,1,1,1,1,1)$ is achievable (for example, $R_{13}$ represents the rate
of the message intended for receivers 1 and 3). Then for any two subsets of receivers ${\cal I} \subset {\cal J}$ we
can perform the operation $R_{\cal I} \rightarrow R_{\cal I} + 1, R_{\cal J} \rightarrow R_{\cal J} - 1$, and for any
two subsets of receivers ${\cal I} \neq {\cal J}$ we can perform the operation $R_{\cal I} \rightarrow R_{\cal I} - 1,
R_{\cal J} \rightarrow R_{\cal J} - 1, R_{{\cal I}\cup {\cal J}} \rightarrow R_{{\cal I}\cup {\cal J}} + 1$. For
instance we may swap the first and second receivers' private bits for a common bit that is sent to the pair, so that
the rate vector $(0,0,1,2,1,1,1)$ is achieved. Similarly the rate vector $(0,1,1,1,1,0,2)$ can be achieved by using the
first receivers private bit and the bit common to the second and third receivers, to send information common to all
three receivers. It can be shown that the number of distinct operations of this form is 15. That is, if the rate vector
$(1,1,1,1,1,1)$ is achievable, so is the set of points contained within a 15-edged polytope, which is the
generalization to $L=3$ of the polytope in figure \ref{fig:BC_example}.

Again we ask the question, is this set of operations complete? Are there any points outside this 15-edged polytope that
are universally achievable on any three-user broadcast channel? The answer, perhaps surprisingly, is yes. There exists
a sixteenth distinct universal encoding/decoding operation. It does not involve a mere relabeling of common and private
bits. It enables the rate vector $(1,1,1,0,0,0,3)$ to be achieved. This new operation together with the fifteen trivial
ones forms the complete set of distinct universal encoding/decoding operations for $L=3$. That is, all other rate
vectors universally achievable from $(1,1,1,1,1,1,1)$ can be achieved by time sharing between these 16 distinct
universal encoding/decoding operations.

Now we turn to the multiple access channel (MAC) with $L$ users. The MAC has also typically been studied in the context
of unicast messaging where it's capacity region has in many cases been completely characterized. For multicasting the
capacity of the discrete memoryless MAC is computed in \cite{SW} and a conjecture regarding the generalization of this
result to an arbitrary number of users is given.

Let us apply the reasoning we applied above for the broadcast channel, to the MAC. Consider a two user MAC. Each
transmitter wishes to send a private message of rate $R_i$ to the receiver for $i\in\{1,2\}$. In addition there is a
common message of rate $R_{12}$ that both transmitters share, and desire to be sent to the receiver. Suppose for a
given MAC a rate vector $(R_1,R_2,R_{12})=(1,1,1)$ is achievable. Then the first transmitter could just label its
rate-one bit stream as common information and send it to the transmitter. Thus the rate vector $(0,1,2)$ is also
achievable. By symmetry the second transmitter could do the same so $(1,0,2)$ is achievable too. Are there any other
operations that tradeoff between elements of the rate vector?\footnote{We could combine these two arguments to conclude
$(0,0,3)$ is achievable but we will not be interested in this operation as it can be expressed as a linear combination
of others.} The answer is no. For the broadcast channel we could swap common information for private, but not so for
the MAC. More specifically we cannot relabel common information as private, as a common bitstream may require both
transmitters have access to it in order for it to be passed to the receiver. A private bitstream assumes only a single
transmitter has access to it. The $(1,1,1)$-multicast region for the two-user MAC is plotted in figure
\ref{fig:MAC_example}. There are three extremal rays and correspondingly three distinct universal/encoding decoding
operations. The first two are stated above and the third consists of merely lowering the common rate so as to arrive at
the point $(1,1,0)$.

Unlike the broadcast channel, this structure directly generalizes to $L$ users. For three users there are ten universal
encoding/decoding operations. Six result from relabeling private information as pairwise. Three result from relabeling
pairwise as common and the last results from lowering the common rate. Thus the multicast capacity region of the
multiple access channel has a less intricate structure than that of the broadcast.

In this paper we characterize the complete set of distinct universal encoding/decoding operations and the associated
region of achievable rate vectors, for both the broadcast channel and the MAC channel, for $L=3$. In essence this is a
characterization of the universal constraints on the multicast capacity region of these channels.

Section II describes the notation we use. In section III we describe the problem in detail. Section IV presents the
results and section V and VI the proofs.

\begin{figure}
\centering
\includegraphics[width=300pt]{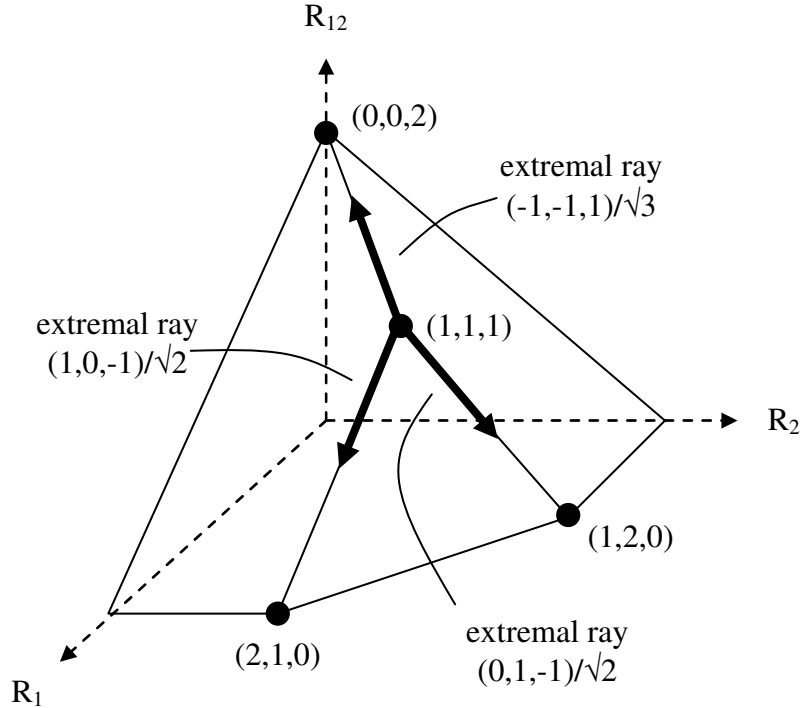}
\caption{The $(1,1,1)$-multicast region for the broadcast channel, $L=2$.} \label{fig:BC_example}
\end{figure}

\begin{figure}
\centering
\includegraphics[width=300pt]{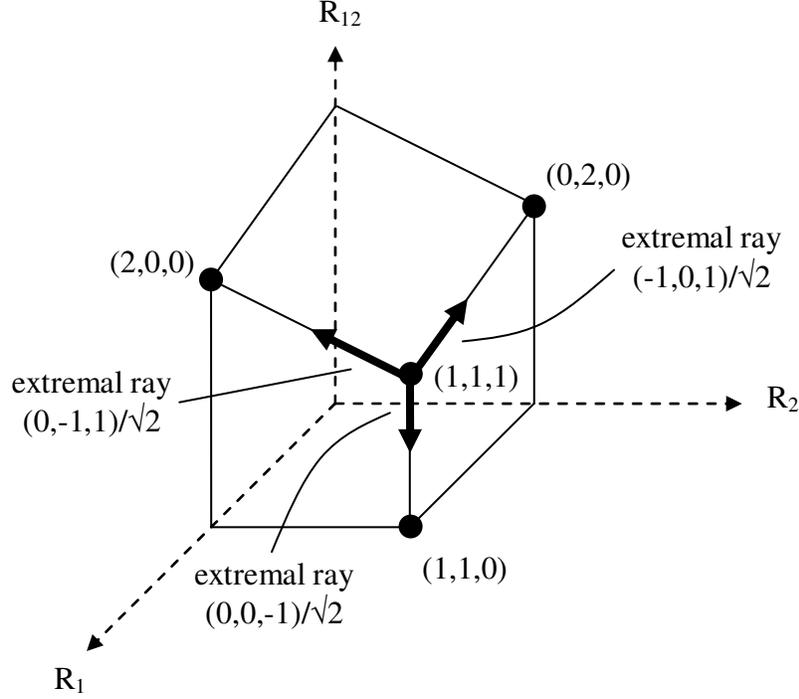}
\caption{The $(1,1,1)$-multicast region for the multiple access channel, $L=2$.} \label{fig:MAC_example}
\end{figure}

\section{Preliminary Notation}

We briefly describe some of the notation that will be used. Typically ${\cal I}$ and ${\cal J}$ will be used to denote
subsets of $\{1,2,3\}$. For example we may have ${\cal I} = \{2,3\}$, which would imply $R_{\cal I} \equiv R_{\{2,3\}}
\equiv R_{23}$. Rates in bold font represent tuples, for example we may have ${\bf R} = (R_1,R_2,R_{12})$. Elements of
time series are indicated by a index in parentheses following the variable, for example $Y(i)$. An entire time series
is represented by bold font, for example ${\bf W}_1 = [W_1(1),\dots,W_1(n)]$ If ${\cal S}$ is a set then $2^S$ denotes
the powerset (the set of all subsets of $S$) excluding the nullset, e.g. if $S = \{1,2\}$ then $2^S \equiv
\{\{1\},\{2\},\{1,2\}\}$. We denote the nullset by $\phi$.  The symbol $\preceq$ denotes element-wise inequality.

\section{Problem Setup}

Consider a broadcast channel with three receivers. The input alphabet is denoted $\cal X$ and the output alphabets
${\cal Y}_1, {\cal Y}_2, {\cal Y}_3$. The probability transition function is $p(y_1,y_2,y_3|x)$. The message vector is
\begin{equation*}
(W_1,W_2,W_3,W_{12},W_{13},W_{23},W_{123}).
\end{equation*}

\noindent The subscript denotes the subset of receivers for which the message in intended, for example message $W_{23}$
is intended for receivers 2 and 3. Denote the rate vector ${\bf R} = (R_1,R_2,R_3,R_{12},R_{13},R_{23},R_{123})$. A
$(2^{n{\bf R}},n)$ code consists of an encoder
\begin{equation*}
x^n: \prod_{{\cal I} \subseteq \{1,2,3\}} \{1,\dots,2^{nR_{\cal I}}\} \rightarrow {\cal X}^n
\end{equation*}

\noindent and twelve decoders
\begin{align*}
&{\hat W}_{i,{\cal I}}: {\cal Y}_i^n \rightarrow 2^{nR_{\cal I}}
\end{align*}

\noindent where $i \in \{1,2,3\}$ denotes the receiver and ${\cal I} \subseteq \{1,2,3\}$ with $i \in \cal I$ denotes
the message index. Thus each receiver decodes four messages (the first receiver decodes $W_1,W_{12},W_{13},W_{123}$,
etc...). The probability of error $P_e^{(n)}$ is defined to be the probability that at least one of the decoded
messages is not equal to the transmitted message, i.e.
\begin{equation*}
P_e^{(n)} = P\left( \bigcup_{\scriptsize \begin{array}{c}
                               {\cal I} \subseteq \{1,2,3\} \\
                               \text{s.t. } i \in {\cal I}
                             \end{array}}
\left\{ {\hat W}_{i,{\cal I}}(Y_i^n) \neq W_{i,{\cal I}} \right\} \right).
\end{equation*}

\noindent where the seven messages are assumed to be mutually independent and uniformly distributed over $\prod_{{\cal
I} \in \{1,2,3\}} \{1,\dots,2^{nR_{\cal I}}\}$.

\begin{figure*}
\centering
\includegraphics[width=450pt]{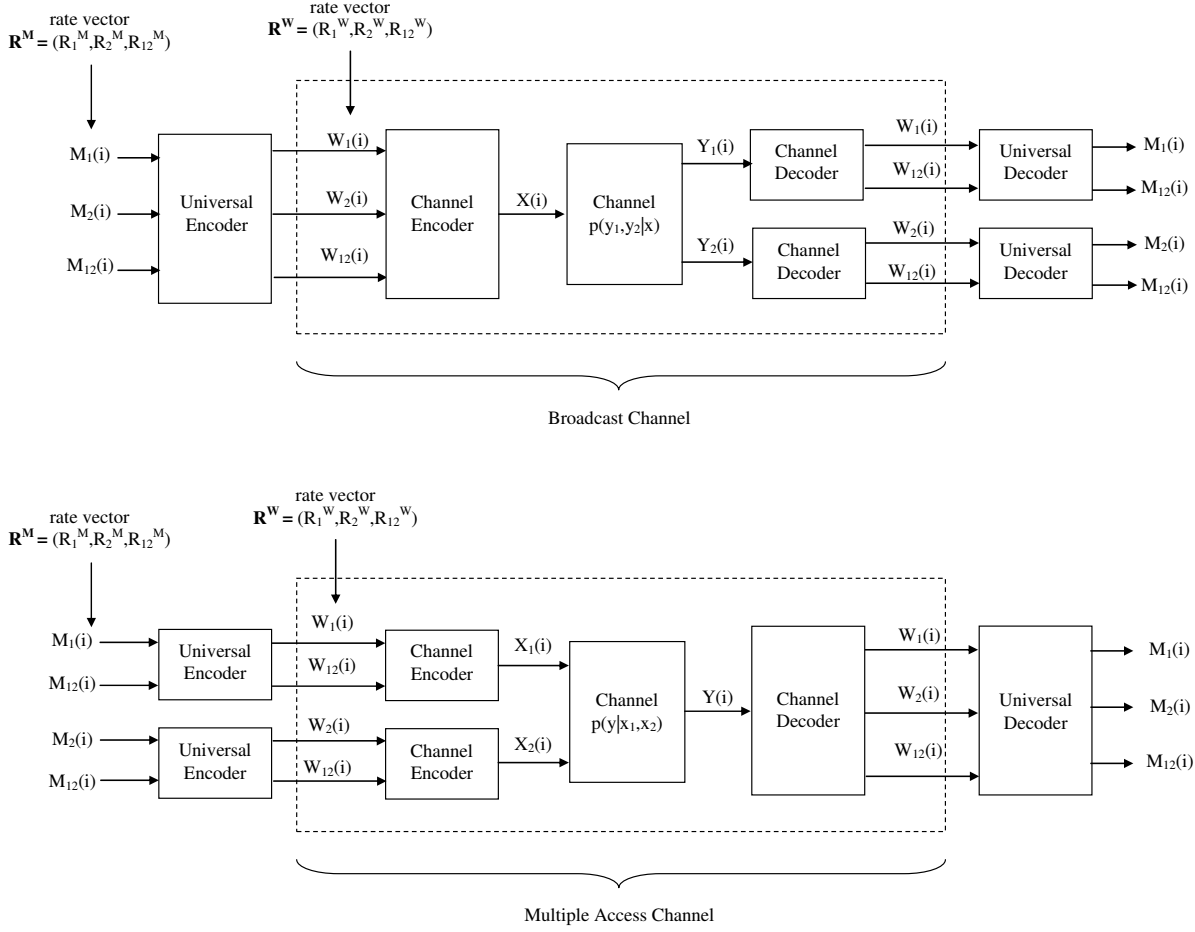}
\caption{System diagrams for $L=2$.}\label{fig:system}
\end{figure*}

\begin{definition}
A multicast rate vector $\bf R$ is said to be achievable for the broadcast channel if there exists a sequence of
$(2^{n{\bf R}},n)$ codes with $P_e^{(n)}\rightarrow 0$.
\end{definition}

\begin{definition}
The multicast capacity region of the broadcast channel is the closure of the set of achievable multicast rate vectors.
It is denoted ${\cal C}_{p(y_1,y_2,y_3|x)}$ or ${\cal C}$ for short.
\end{definition}

\noindent Often we will omit the adjective `multicast'.

We now give a defintion that makes precise the operation of swapping common and private messages, and quantifies the
change in the rate vector. Let ${\bf R}^W$ and ${\bf R}^M$ be two rate vectors.

\begin{definition}
A $(d{\bf R},n)$-universal encoding/decoding operation is a pair of mappings
\begin{align*}
\begin{array}{c}
  W_{\cal J} : \displaystyle \prod_{{\cal I} \subseteq \{1,2,3\}} \{1,\dots,2^{nR_{\cal I}^M}\} \rightarrow \{1,\dots,2^{nR_{\cal
J}^W}\}, \text{and} \\ \\
  {\hat M}_{i,{\cal I}} : \displaystyle \prod_{\scriptsize \begin{array}{c}
                               {\cal J} \subseteq \{1,2,3\} \\
                               \text{s.t. } i \in {\cal J}
                             \end{array}}
\{1,\dots,2^{nR_{\cal J}^W}\} \rightarrow \{1,\dots,2^{nR_{\cal I}^M}\}
\end{array}
\end{align*}

\noindent for all ${\cal J} \subseteq \{1,2,3\}$ and all $i \in \{1,2,3\}$ for all ${\cal I} \subseteq \{1,2,3\}$ such
that $i \in \cal I$, with the properties ${\hat M}_{i,{\cal I}} = {\hat M}_{j,{\cal I}}$ for all $i,j \in \{1,2,3\}$,
${\bf R}^M \neq {\bf R}^W$ and
\begin{equation*}
\frac{{\bf R}^M-{\bf R}^W}{\lVert {\bf R}^M-{\bf R}^W \rVert} = d{\bf R},
\end{equation*}
\noindent $W(M)$ being the universal encoder output and ${\hat M}(\hat W)$ being the universal decoder output. The
vector $d{\bf R}$ is referred to as the `normalized difference vector'.
\end{definition}

\noindent The property ${\hat M}_{i,{\cal I}} = {\hat M}_{j,{\cal I}}$ for all $i,j \in \{1,2,3\}$ ensures that all
users agree on the common messages they decode. See figure \ref{fig:system} for a system diagram that illustrates the
relationship between $M,W,{\hat M}$ and ${\hat W}$.

\begin{example}\label{exm:universal encoding}
Suppose ${\bf R}^W = (1,0,0,1,0,0,0)$ and ${\bf R}^M = (2,0,0,0,0,0,0)$. Let $n=1$. Then the mapping $W_1({\bf M}) =
M_1(1)$, $W_{12}({\bf M}) = M_1(2)$ is a universal encoding operation with $d{\bf R} = (1,0,0,-1,0,0,0)/\sqrt{2}$. The
universal decoding operation is the inverse mapping given by ${\hat M}_1({\hat {\bf W}}) = [{\hat W}_1, {\hat
W}_{12}]$.
\end{example}

\begin{definition}
A $d{\bf R}$-universal encoding/decoding operation is called `distinct' if the vector $d{\bf R}$ cannot be expressed as
positive linear combination of vectors $\{d{\bf R}_i\} \neq d{\bf R}$ for which there exist $d{\bf R}_i$-universal
encoding/decoding operations for $i=1,2,\dots$. The (rays associated with the) normalized difference vectors
corresponding to distinct $d{\bf R}$-universal encoding/decoding operations are called `extremal rays'.
\end{definition}

\noindent By positive linear combination we mean a weighted linear sum with non-negative coefficients.

\begin{example}
It will be evident later that the universal encoding/decoding operation given in example \ref{exm:universal encoding}
is distinct and thus $(1,0,0,-1,0,0,0)/\sqrt{2}$ is an extremal ray. By symmetry $(0,1,0,-1,0,0,0)/\sqrt{2}$ is also an
extremal ray. Note distinctness does not imply uniqueness --the universal encoding/decoding operation that moves from
rate vector ${\bf R}^W = (1,0,0,1,0,0,0)$ to rate vector ${\bf R}^M = (1.5,0,0,0.5,0,0,0)$ is also classified as
distinct, but it has the same normalized difference vector. An example of a universal encoding/decoding operation that
is not distinct is one that moves from rate vector ${\bf R}^W = (1,0,0,1,0,0,0)$ to rate vector ${\bf R}^M =
(1.5,0.5,0,0,0,0,0)$. Denote the corresponding normalized difference vector is $d{\bf R}_A \triangleq
(0.5,0.5,0,-1,0,0,0)/\sqrt{1.5}$. The universal encoding/decoding operations that achieve this shift correspond to
time-sharing between two operations, one with normalized difference vector $d{\bf R}_B \triangleq
(1,0,0,-1,0,0,0)/\sqrt{2}$, the other with normalized difference vector $d{\bf R}_C \triangleq
(0,1,0,-1,0,0,0)/\sqrt{2}$. Indeed we have
\begin{equation*}
d{\bf R}_A = \frac{1}{\sqrt{3}}d{\bf R}_B + \frac{1}{\sqrt{3}}d{\bf R}_C.
\end{equation*}
\end{example}

\noindent We now give a formal definition of the region alluded to in figure \ref{fig:BC_example}. Let
\begin{equation*}
{\bf R}^* = (R_1^*,R_2^*,R_3^*,R_{12}^*,R_{13}^*,R_{23}^*,R_{123}^*)
\end{equation*}

\noindent be a parameter.

\begin{definition}
The `${\bf R}^*$-multicast region' is the intersection of the capacity regions of all broadcast channels for which the
rate vector ${\bf R}^*$ is achievable, i.e.
\begin{equation*}
\bigcap_{p(y_1,y_2,y_3|x) : {\bf R}^* \in {\cal C}_{p(y_1,y_2,y_3|x)}} {\cal C}_{p(y_1,y_2,y_3|x)}
\end{equation*}
\end{definition}

\noindent See figures \ref{fig:BC_example} for examples of this region.

As the problem setup for the multiple access channel is entirely analogous to the aforementioned setup for the
broadcast channel, we do not explicitly describe it. An example of the ${\bf R}^*$-multicast region is given in figure
\ref{fig:MAC_example}

The aim of this paper is to characterize the ${\bf R}^*$-multicast cones for both the broadcast and multiple access
channels.

\section{Results}

\begin{figure*} \label{fig:matricesL2}
\begin{equation*}
{\bf G}_{BC,2} = \left[%
\begin{array}{rrr}
  1 & 0 & 1 \\
  0 & 1 & 1 \\
  1 & 1 & 1 \\
\end{array}
\right] \quad\quad
{\bf H}_{BC,2} = \left[%
\begin{array}{rrr}
  1 & 0 & -1 \\
  0 & 1 & -1 \\
  -1 & -1 & 1 \\
\end{array}%
\right]
\end{equation*}
\begin{equation*}
{\bf G}_{MAC,2} = \left[%
\begin{array}{rrr}
  1 & 0 & 1 \\
  0 & 1 & 1 \\
  1 & 1 & 1 \\
\end{array}
\right]
\quad\quad {\bf H}_{MAC,2} = \left[
\begin{array}{rrr}
  1 & 0 & -1 \\
  0 & 1 & -1 \\
  -1 & -1 & 1 \\
\end{array}%
\right]
\end{equation*}
\caption{Results for $L = 2$.}
\end{figure*} \label{fig:matricesL3}
\begin{figure*}
\begin{equation*}
{\bf G}_{BC,3} = \left[%
\begin{array}{rrrrrrrrrrrrrrr}
  1 & 0 & 0 & 1 & 1 & 0 & 1 & 1 & 1 & 1 & 1 & 2 & 2 & 1 & 2 \\
  0 & 1 & 0 & 1 & 0 & 1 & 1 & 1 & 1 & 1 & 1 & 2 & 1 & 2 & 2 \\
  0 & 0 & 1 & 0 & 1 & 1 & 1 & 1 & 1 & 1 & 1 & 1 & 2 & 2 & 2 \\
  1 & 1 & 0 & 1 & 1 & 1 & 1 & 2 & 1 & 1 & 2 & 2 & 2 & 2 & 2 \\
  1 & 0 & 1 & 1 & 1 & 1 & 1 & 1 & 2 & 1 & 2 & 2 & 2 & 2 & 2 \\
  0 & 1 & 1 & 1 & 1 & 1 & 1 & 1 & 1 & 2 & 2 & 2 & 2 & 2 & 2 \\
  1 & 1 & 1 & 1 & 1 & 1 & 1 & 2 & 2 & 2 & 2 & 3 & 3 & 3 & 3 \\
\end{array}
\right]
\end{equation*}
\begin{equation*}
{\bf H}_{BC,3} = \left[%
\begin{array}{rrrrrrrrrrrrrrrr}
  -1 & -1 & 0 & 1 & 0 & 0 & 1 & 0 & 0 & -1 & 0 & 0 & 0 & 0 & 0 & 0 \\
  -1 & 0 & -1 & 0 & 1 & 0 & 0 & -1 & 0 & 0 & 1 & 0 & 0 & 0 & 0 & 0 \\
  0 & -1 & -1 & 0 & 0 & -1 & 0 & 0 & 1 & 0 & 0 & 1 & 0 & 0 & 0 & 0 \\
  1 & 0 & 0 & -1 & -1 & -1 & 0 & 0 & 0 & 0 & 0 & 0 & 1 & 0 & 0 & -1 \\
  0 & 1 & 0 & 0 & 0 & 0 & -1 & -1 & -1 & 0 & 0 & 0 & 0 & 1 & 0 & -1 \\
  0 & 0 & 1 & 0 & 0 & 0 & 0 & 0 & 0 & -1 & -1 & -1 & 0 & 0 & 1 & -1 \\
  0 & 0 & 0 & 0 & 0 & 1 & 0 & 1 & 0 & 1 & 0 & 0 & -1 & -1 & -1 & 2 \\
\end{array}
\right]
\end{equation*}
\begin{equation*}
{\bf G}_{MAC,3} = \left[%
\begin{array}{rrrrrrrrrrrrrrr}
  1 & 0 & 0 & 1 & 1 & 0 & 1 & 1 & 1 & 1 & 1 \\
  0 & 1 & 0 & 1 & 0 & 1 & 1 & 1 & 1 & 1 & 1 \\
  0 & 0 & 1 & 0 & 1 & 1 & 1 & 1 & 1 & 1 & 1 \\
  0 & 0 & 0 & 1 & 0 & 0 & 0 & 1 & 1 & 1 & 1 \\
  0 & 0 & 0 & 0 & 1 & 0 & 1 & 0 & 1 & 1 & 1 \\
  0 & 0 & 0 & 0 & 0 & 1 & 1 & 1 & 0 & 1 & 1 \\
  0 & 0 & 0 & 0 & 0 & 0 & 0 & 0 & 0 & 0 & 1 \\
\end{array}
\right]
\end{equation*}
\begin{equation*}
{\bf H}_{MAC,3} = \left[%
\begin{array}{rrrrrrrrrrrrrrrr}
  1 & 1 & 0 & 0 & 0 & 0 & 0 & 0 & 0 & 0 \\
  0 & 0 & 1 & 1 & 0 & 0 & 0 & 0 & 0 & 0 \\
  0 & 0 & 0 & 0 & 1 & 1 & 0 & 0 & 0 & 0 \\
  -1 & 0 & -1 & 0 & 0 & 0 & 1 & 0 & 0 & 0 \\
  0 & -1 & 0 & 0 & -1 & 0 & 0 & 1 & 0 & 0 \\
  0 & 0 & 0 & -1 & 0 & -1 & 0 & 0 & 1 & 0 \\
  0 & 0 & 0 & 0 & 0 & 0 & -1 & -1 & -1 & 1 \\
\end{array}
\right]
\end{equation*}
\caption{Results for $L = 3$.}
\end{figure*}

\begin{theorem}\label{thm:main_result_BC}
For $L=3$ the ${\bf R}^*$-multicast region of the broadcast channel is the set of all ${\bf R} \in {\mathbb R}_+^7$
satisfying
\begin{equation}\label{eqn:BC_region}
{\bf G}_{BC,3}^T\left({\bf R}-{\bf R}^*\right) \preceq 0
\end{equation}
where ${\bf G}_{BC,3}$ is given in figure \ref{fig:matricesL3}. This region is a polytope, characterized by the cone
$\{{\bf R}\in{\mathbb R}^7 : {\bf G}_{BC,3}^T{\bf R} \preceq 0\}$. We refer to this cone as the $L=3$ 'multicast cone'.
The sixteen extremal rays of this cone are given by the columns of the matrix ${\bf H}_{BC,3}$ in figure
\ref{fig:matricesL3}. Thus there are 16 distinct universal encoding/decoding operations for $L=3$.
\end{theorem}

The $(1,1,1)$-multicast region for the broadcast channel for $L=2$ is illustrated in figure \ref{fig:BC_example}. For
$L=2$ there are 3 distinct universal encoding/decoding operations. The ${\bf G}_{BC,2}$ and ${\bf H}_{BC,2}$ matrices
are given in figure \ref{fig:matricesL2}.

\noindent The columns of ${\bf G}_{BC,2}$ are the normal vectors to the three hyperplanes bounding the region. The
columns of ${\bf H}_{BC,2}$ are the three extremal rays (see figure \ref{fig:BC_example}).

\begin{theorem}\label{thm:main_result_MAC}
For $L=3$ the ${\bf R}^*$-multicast region of the multiple access channel is the set of all ${\bf R} \in {\mathbb
R}_+^7$ satisfying
\begin{equation}\label{eqn:MAC_region}
{\bf G}_{MAC,3}^T\left({\bf R}-{\bf R}^*\right) \preceq 0
\end{equation}
where ${\bf G}_{MAC,3}$ is given in figure \ref{fig:matricesL3}. This region is also a polytope characterized by the
cone $\{{\bf R}\in{\mathbb R}^7 : {\bf G}_{MAC,3}^T{\bf R} \preceq 0\}$. The 10 extremal rays of this cone are given by
the columns of the matrix ${\bf H}_{MAC,3}$ in figure \ref{fig:matricesL3}. Thus there are 10 distinct universal
encoding/decoding operations for $L=3$.
\end{theorem}

The $(1,1,1)$-multicast region for the MAC for $L=2$ is illustrated in figure \ref{fig:MAC_example}. There are 3
distinct universal encoding/decoding operations. The ${\bf G}_{MAC,2}$ and ${\bf H}_{MAC,2}$ matrices are given in
figure \ref{fig:matricesL2}.

An alternative interpretation of theorem \ref{thm:main_result_BC} is the following (the same interpretation applies for
\ref{thm:main_result_MAC}). For notational simplicity we denote the capacity region of an arbitrary broadcast channel
by $\cal C$. Let
\begin{equation*}
{\bf R}^*(\alpha) = \text{arg}\max_{R \in {\cal C}} \alpha^T {\bf R}
\end{equation*}

\noindent ${\bf R}^*(\alpha)$ is the rate vector lying on the boundary of the capacity region in the direction of
$\alpha$. Let
\begin{equation*}
{\cal C}^*(\alpha) = \left\{ {\bf R} \in {\mathbb R}_+^7 \left| \alpha^T{\bf R} \le \alpha^T{\bf R}^*(\alpha) \right.
\right\}.
\end{equation*}

\noindent ${\cal C}^*(\alpha)$ is the halfspace of all rate vectors lying underneath the hyperplane $\alpha^T{\bf R} =
\alpha^T{\bf R}^*(\alpha)$. The region $\cal C$ is convex and thus we can characterize it by its support function
${\cal C}^*(\alpha)$, i.e.
\begin{equation*}
{\cal C} = \bigcap_{\alpha \in {\mathbb R}_+^7} {\cal C}(\alpha).
\end{equation*}

\noindent However this is not the minimal dual representation of ${\cal C}$. Let
\begin{equation*}
{\cal H}^*_3 \triangleq \left\{ \alpha \in {\mathbb R}_+^7 \left| \alpha^T {\bf H}_{BC,3} \preceq 0 \right. \right\}
\end{equation*}

\begin{corollary}
The multicast capacity region of an any broadcast channel with three receivers can be expressed as
\begin{equation*}
{\cal C} = \bigcap_{\alpha \in {\cal H}} {\cal C}(\alpha).
\end{equation*}

\noindent if and only if
\begin{equation*}
{\cal H} \supseteq {\cal H}^*_3.
\end{equation*}
\end{corollary}

\noindent This says the following: when computing the multicast capacity region of a broadcast channel by maximizing
the weighted sum-rate, the smallest set that one need vary the weighting coefficients $\alpha$ over is ${\cal H}^*_3$.
Put another way, the normal vector $\alpha$ to any point on the boundary of the multicast capacity region is always
contained in the set ${\cal H}^*_3$. See figure \ref{fig:valid_region}.

\begin{figure}
\centering
\includegraphics[width=450pt]{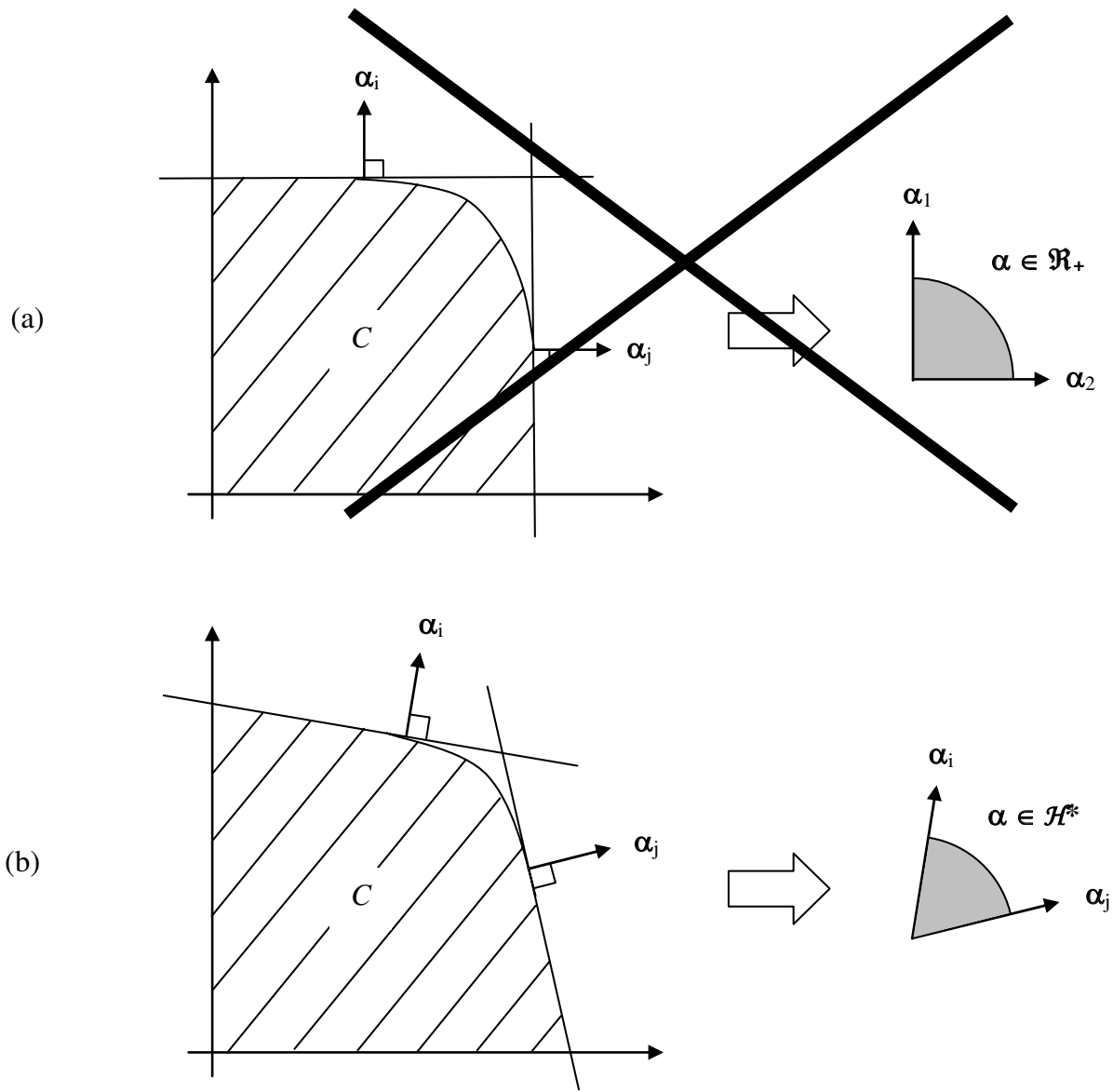}
\caption{The normal vector $\alpha$ of the broadcast channel capacity region satisfies $\alpha^T \in {\cal H}^*$. (a) A
capacity region that cannot occur. (b) A capacity region that can occur.} \label{fig:valid_region}
\end{figure}

\section{Proof of Theorem \ref{thm:main_result_BC}}
The direct part of the proof consists of showing that for any broadcast channel, if a rate vector ${\bf R}^*$ is
achievable then all rate vectors in the region given by equation (1) are achievable. This establishes that the ${\bf
R}^*$-multicast region is `at least as large' as the region given by equation (1). The converse part of the proof
consists of illustrating, for each ${\bf R}^* \in {\mathbb R}_+^7$, a broadcast channel for which no rate vector
outside the region given by equation (1) is achievable. This establishes that the ${\bf R}^*$-multicast region is `at
least as small' as the region given by equation (1). We start with the direct part. For notational simplicity we drop
the broadcast channel ({\it BC}) subscript.

\subsection{Direct Part}
Suppose that ${\bf R}^*$ is achievable for a particular broadcast channel. We show that any rate-vector ${\bf R} \in
{\mathbb R}_+^7$ satisfying
\begin{equation}
{\bf R} \preceq {\bf R}^* + {\bf H}_{BC,3}\Delta
\end{equation}

\noindent for $\Delta \in {\mathbb R}_+^{16}$ is also achievable. We then show that this region is precisely the one
given in equation (1). Let $\Delta_i$ denote the $i$th element of $\Delta$ and ${\bf H}_{BC,3}(i)$ denote the $i$th
column of ${\bf H}_{BC,3}$. To show that any rate-vector satisfying equation (2) is achievable, we show that each of
the 16 rate-vectors given by
\begin{equation}
{\bf R}^{(i)} = {\bf R}^* + {\bf H}_{BC,3}(i)\Delta_i^*, \quad\quad i=1,\dots,16
\end{equation}

\noindent are achievable where
\begin{equation*}
\Delta_i^* = \max_{{\bf H}_{BC,3}(i)\Delta_ i\preceq {\bf R}^*} \Delta
\end{equation*}
By time sharing between these vectors the entire boundary region $\{ {\bf R}^* - {\bf H}_{BC,3}\Delta | \Delta \in
{\mathbb R}_+^{16}\}$ is achieved and hence any point within it (i.e. satisfying equation (2)) can also be achieved.

Let ${\bf M}, {\hat {\bf M}}$ correspond to the binary message vector and estimate of the message vector, respectively,
that the transmitter wishes to send at rate vector ${\bf R}^{(i)}$. We illustrate the achievability of equation (3) for
$i=3$.

To universally encode for $i=3$, assume without loss of generality that $R_1^*\le R_2^*$. In what follows we ignore
rounding effects as it will be clear that in the limit $n \rightarrow \infty$ they are negligible. Set
\begin{align*}
W_1^n &= [M_{12}(1),\dots,M_{12}(nR_1^*)] \\
W_2^n &= [M_{12}(1),\dots,M_{12}(nR_1^*), M_{2}(1),\dots,M_{2}(nR_2^*-nR_1^*)] \\
W_{12}^n &= [M_{12}(nR_1^*+1),\dots,M_{12}(nR_{1}^*+nR_{12}^*)]
\end{align*}

\noindent In words, the information common to receivers 1 and 2 is split into two parts. The first part is replicated
and sent separately down both receiver 1 and receiver 2's private channels. The second part is sent down the channel
common to both receivers. As receiver 2's private channel can accommodate a higher bit-rate than receiver 1's, there is
some bandwidth left over. This is allocated to sending some of receiver 2's private information.

For all other subsets ${\cal I}$ of $\{1,2,3\}$ set $W_{\cal I}^n=M_{\cal I}^n$ and ${\hat M}_{\cal I}^n = {\hat
W}_{\cal I}^n$. Universal decoding is straightforward. The first receiver sets
\begin{align*}
{\hat M}_{1,12}^n &= [{\hat W}_1^n, {\hat W}_{12}^n] \\
{\hat M}_{1,13}^n &= {\hat W}_{13}^n \\
{\hat M}_{1,123}^n &= {\hat W}_{123}^n
\end{align*}

\noindent and in this way successfully recovers its message, as the achievability of ${\bf R}^*$ implies that $\bf W$
was decoded correctly. The second receiver sets
\begin{align*}
{\hat M}_{2,2}^n &= [{\hat W}_2(nR_1^*+1), \dots, {\hat W}_2(nR_2^*)] \\
{\hat M}_{2,12}^n &= [{\hat W}_2(1),\dots,{\hat W}_2(nR_1^*), {\hat W}_{12}^n] \\
{\hat M}_{2,23}^n &= {\hat W}_{23}^n \\
{\hat M}_{2,123}^n &= {\hat W}_{123}^n
\end{align*}

\noindent and is similarly successful in decoding. The third receivers sets ${\hat M}_{\cal I}^n = {\hat W}_{\cal I}^n$
for all of its messages. Then we have achieved a rate vector of
\begin{align*}
{\bf R}^{(3)} &= {\bf R}^* + \left[ \begin{array}{c}
                              -1 \\
                              -1 \\
                              0 \\
                              1 \\
                              0 \\
                              0 \\
                              0
                            \end{array} \right] R_1^* \\
&= {\bf R}^* + {\bf H}_{BC,3}(3)\Delta_3.
\end{align*}
with $\Delta_3 = R_1^*$. The universal encoding and decoding procedures for all other $i\in\{1,\dots,15\}$ are similar
and follow from the structure of the columns of the matrix ${\bf H}_{BC,3}$.

Universal encoding and decoding for $i=16$ is different. Assume without loss of generality that $R_{12}^* \le R_{13}^*
\le R_{23}^*$. To encode, set $W_i^n=M_i^n$ for $i=1,2,3$ and
\begin{align*}
W_{12}^n &= [M_{123}(1),\dots,M_{123}(nR_{12}^*)] \\
W_{13}^n &= [M_{123}(nR_{12}^*+1),\dots,M_{123}(2nR_{12}^*),M_{13}(1),M_{13}(nR_{13}-nR_{12})] \\
W_{23}^n &= [M_{123}(1)\oplus M_{123}(nR_{12}^*+1),\dots, M_{123}(nR_{12}^*) \oplus M_{123}(2nR_{12}^*), \\
&\quad\quad\quad M_{23}(1),\dots,M_{23}(nR_{23}-nR_{12})] \\
W_{123}^n &= [M_{123}^n(2nR_{12}^*),\dots,M_{123}^n(2nR_{12}^*+nR_{123}^*)]
\end{align*}

\noindent In words, the information common to all receivers is split into three streams. The first and second are sent
at rate $R_{12}^*$ using the three pairwise links. The third stream is sent at rate $R_{123}^*$ across the link common
to all receivers.

The first receiver decodes by setting ${\hat M}_{1,1}^n={\hat W}_1^n$ and
\begin{align*}
M_{13}^n &= [W_{13}(nR_{12}^*+1),\dots,W_{13}(nR_{13}^*)] \\
M_{123}^n &= [W_{12}^n, W_{13}^n, W_{123}^n].
\end{align*}

\noindent The second receiver decodes by setting ${\hat M}_{1,2}^n={\hat W}_2^n$ and
\begin{align*}
M_{23}^n &= [W_{23}(nR_{12}^*+1),\dots,W_{23}(nR_{23}^*)] \\
M_{123}^n &= [W_{12}^n, W_{12}^n\oplus W_{23}^n, W_{123}^n].
\end{align*}

\noindent The third receiver decodes by setting ${\hat M}_{1,3}^n={\hat W}_3^n$ and
\begin{align*}
M_{13}^n &= [W_{13}(nR_{12}^*+1),\dots,W_{13}(nR_{13}^*)] \\
M_{23}^n &= [W_{23}(nR_{12}^*+1),\dots,W_{23}(nR_{23}^*)] \\
M_{123}^n &= [W_{13}^n, W_{13}^n\oplus W_{23}^n, W_{123}^n].
\end{align*}

\noindent Then we have achieved a rate vector of
\begin{align*}
{\bf R}^{(3)} &= {\bf R}^* + \left[ \begin{array}{c}
                              0 \\
                              0 \\
                              0 \\
                              -1 \\
                              -1 \\
                              -1 \\
                              2
                            \end{array} \right] R_{12}^* \\
&= {\bf R}^* + {\bf H}_{BC,3}(16)\Delta_{16}.
\end{align*}
with $\Delta_{16} = R_{12}^*$. Thus the 16 rate vectors satisfying equation (3) are achievable and by time sharing
between them, all rate vectors in the region given by equation (2) are achievable.

It remains to show that this region is equivalent to the one in equation (1), i.e. that for any ${\bf R}^* \in {\mathbb
R}_+^7$
\begin{equation*}
\left\{ \left. {\bf R} \in {\mathbb R}_+^7 \right| {\bf G}_{BC,3}^T{\bf R} \preceq {\bf G}_{BC,3}^T{\bf R} \right\}
\equiv \left\{ \left. {\bf R} \in {\mathbb R}_+^7 \right| {\bf R} \preceq {\bf R}^* + {\bf H}_{BC,3}\Delta, \forall
{\Delta} \in {\mathbb R}_+^{16} \right\}.
\end{equation*}

\noindent On the left is the characterization of the polytope in terms of the hyperplanes bounding it. On the right is
the dual characterization in terms of the edges of the polytope (1-dimensional facets). This equivalence can be
demonstrated using computer software such as {\it polymake}.

\subsection{Converse Part}
To establish the converse we now present, for each ${\bf R}^*$, a particular (deterministic) broadcast channel and show
its capacity region is equal to (1). Let the input alphabet ${\cal X} = \prod_{{\cal I} \subseteq \{1,2,3\}}
\{0,\dots,2^{nR_{\cal I}^*}-1 \}$ with the $i$th channel input
\begin{equation*}
X(i) = [X_1(i),X_2(i),X_3(i), X_{12}(i),X_{13}(i),X_{23}(i),X_{123}(i)]
\end{equation*}

\noindent so that each $X_{\cal I} \in \{0,\dots,2^{nR_{\cal I}^*}-1 \}$, and let ${\cal Y}_i \in \prod_{{\cal I}
\subseteq \{1,2,3\}, i \in I} \{0,\dots,2^{nR_{\cal I}^*}-1 \}$ for $i=1,2,3$ with
\begin{align*}
Y_1(i) &= [X_1(i),X_{12}(i),X_{13}(i),X_{123}(i)] \\
Y_2(i) &= [X_2(i),X_{12}(i),X_{23}(i),X_{123}(i)] \\
Y_3(i) &= [X_3(i),X_{13}(i),X_{23}(i),X_{123}(i)].
\end{align*}

\begin{figure}
\centering
\includegraphics[width=450pt]{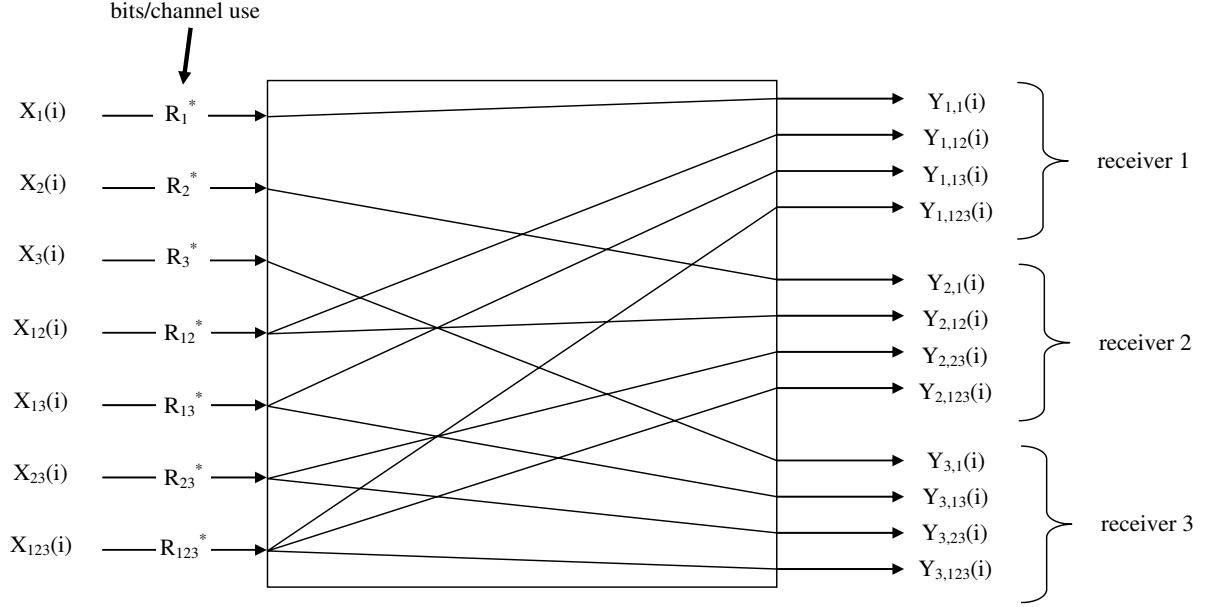}
\caption{Illustration of the deterministic broadcast channel used in converse.}\label{fig:specific_BC}
\end{figure}

\noindent See figure \ref{fig:specific_BC} for an illustration of the channel. Suppose the channel is used $n$ times.
The messages to be transmitted are $W_{\cal I}\sim U(\{1,\dots,2^{nR_{\cal I}}\})$ and mutually independent. Denote the
length-$n$ vector of channel inputs by ${\bf X}$ and the length-$n$ vectors of channel outputs by ${\bf Y}_1, {\bf
Y}_2$ and ${\bf Y}_3$. Let ${\bf G}_{BC,3}(i)$ denote the $i$th column of ${\bf G}_{BC,3}$. We wish to show
\begin{equation}
{\bf G}_{BC,3}(i)^T{\bf R} \leq {\bf G}_{BC,3}(i)^T{\bf R}^*
\end{equation}

\noindent for $i = 1,\dots,15$. Before this we introduce some notation. Suppose $\cal A$ is a collection of subsets of
$\{1,2,3\}$, for example ${\cal A} = \{1,2,12,13,123\}$. The collection $\cal A$ should be thought as the indices of a
subset of the seven channel links (see figure \ref{fig:specific_BC}), for example ${\cal A} = \{1,123\}$ corresponds to
two links, the private one from $X_1$ to $Y_{1,1}$ and the common one from $X_{123}$ to
$Y_{1,123},Y_{2,123},Y_{3,123}$. By $\lfloor {\cal A} \rfloor$ we denote the indices of the messages intended for those
receivers cut by ${\cal A}$. For example if ${\cal A} = \{1,2,12,13,123\}$ then all links to the first receiver are
cut, but not all links to the second or the third. As the first receiver is sent the messages $W_1,W_{12},W_{13}$ and
$W_{123}$, we have $\lfloor {\cal A} \rfloor = \{1,12,13,123\}$. As another example let ${\cal A} =
\{2,3,12,13,23,123\}$. Then all links to both the second and third receivers are cut and $\lfloor {\cal A} \rfloor =
\{2,3,12,13,23,123\} = {\cal A}$. As a final example if ${\cal A} = \{1,12,23,123\}$ no receivers are completely cut,
and thus $\lfloor A \rfloor = \phi$.

\begin{lemma} Let ${\cal A}_1,{\cal A}_2$ and ${\cal A}_3$ be three collections of subsets of $\{1,2,3\}$ such that
either ${\cal A}_1 \subseteq {\cal A}_2 \cup {\cal A}_3$, ${\cal A}_2 \subseteq {\cal A}_1 \cup {\cal A}_3$ or ${\cal
A}_3 \subseteq {\cal A}_1 \cup {\cal A}_2$. Then
\begin{multline*}
\sum_{{\cal I}\in \lfloor {\cal A}_1 \cup {\cal A}_2 \cup {\cal A}_3 \rfloor} R_{\cal I} + \sum_{
\scriptsize
\begin{array}{c}
  {\cal I}\in \lfloor {\cal A}_1 \cup {\cal A}_2 \rfloor \cap \\
  \lfloor {\cal A}_1 \cup {\cal A}_3 \rfloor \cap \lfloor {\cal A}_2 \cup {\cal A}_3 \rfloor
\end{array}}
R_{\cal I} + \sum_{ \scriptsize
\begin{array}{c}
  {\cal I}\in \lfloor {\cal A}_1 \rfloor \cap \\
  \lfloor {\cal A}_2 \rfloor \cap \lfloor {\cal A}_3 \rfloor
\end{array}}
R_{\cal I} \; \; \le \; \sum_{{\cal I}\in {\cal A}_1} R_{\cal I}^* + \sum_{{\cal I}\in {\cal A}_2} R_{\cal I}^* +
\sum_{{\cal I}\in {\cal A}_3} R_{\cal I}^*
\end{multline*}
\end{lemma}

\noindent This lemma is a generalization of the cutset bounds to multiple subsets of cuts. Indeed if we set ${\cal A}_2
= \phi$ and ${\cal A}_3 = \phi$ we are left with
\begin{equation*}
\sum_{{\cal I}\in \lfloor {\cal A}_1 \rfloor} R_{\cal I} \le \sum_{{\cal I}\in {\cal A}_1} R_{\cal I}^*
\end{equation*}

\noindent which are precisely the cutset bounds.

\begin{proof}
\begin{align*}
&n\left( \sum_{{\cal I}\in {\cal A}_1} R_{\cal I}^* + \sum_{{\cal I}\in {\cal A}_2} R_{\cal I}^* + \sum_{{\cal I}\in
{\cal A}_3} R_{\cal I}^* \right) \\
&\quad\quad \ge \sum_{{\cal I}\in {\cal A}_1} H({\bf X}_{\cal I}) + \sum_{{\cal I}\in {\cal A}_2}
H({\bf X}_{\cal I}) + \sum_{{\cal I}\in {\cal A}_3} H({\bf X}_{\cal I}) \\
&\quad\quad\ge H\left(\cup_{{\cal I}\in {\cal A}_1} {\bf X}_{\cal I}\right) + H\left(\cup_{{\cal I}\in {\cal A}_2} {\bf
X}_{\cal
I}\right) + H\left(\cup_{{\cal I}\in {\cal A}_3} {\bf X}_{\cal I}\right) \\
&\quad\quad = H\left(\cup_{{\cal I}\in {\cal A}_1 \cup {\cal A}_2 \cup {\cal A}_3} {\bf X}_{\cal I} \right) +
I\left(\cup_{{\cal I}\in {\cal A}_1 \cup {\cal A}_2} {\bf X}_{\cal I}; \cup_{{\cal I}\in {\cal A}_1 \cup {\cal A}_3}
{\bf X}_{\cal I}; \cup_{{\cal I}\in {\cal A}_2 \in {\cal A}_3} {\bf X}_{\cal I} \right)\\
&\quad\quad\quad\quad + I\left(\cup_{{\cal I}\in {\cal A}_1} {\bf X}_{\cal I}; \cup_{{\cal I}\in {\cal A}_2} {\bf
X}_{\cal I}; \cup_{{\cal I}\in {\cal A}_3} {\bf X}_{\cal I} \right) \\
&\quad\quad\ge H\left( \cup_{{\cal I}\in \lfloor {\cal A}_1 \cup {\cal A}_2 \cup {\cal A}_3 \rfloor } {\bf W}_{\cal I}
\right) + H\left( \cup_{{\cal I}\in \lfloor {\cal A}_1 \cup {\cal A}_2 \rfloor \cap \lfloor {\cal A}_1 \cup {\cal A}_3
\rfloor \cap \lfloor {\cal A}_2 \cup {\cal A}_3 \rfloor } {\bf W}_{\cal I} \right) \\
&\quad\quad\quad\quad + H\left( \cup_{{\cal I}\in \lfloor {\cal A}_1 \rfloor \cap \lfloor {\cal A}_2 \rfloor \cap
\lfloor {\cal A}_3 \rfloor } {\bf W}_{\cal I} \right) + \epsilon_n \\
&\quad\quad = n\Bigg( \sum_{{\cal I}\in \lfloor {\cal A}_1 \cup {\cal A}_2 \cup {\cal A}_3 \rfloor} R_{\cal I} +
\sum_{{\cal I}\in \lfloor {\cal A}_1 \cup {\cal A}_2 \rfloor \cap \lfloor {\cal A}_1 \cup {\cal A}_3 \rfloor \cap
\lfloor {\cal A}_2 \cup {\cal A}_3 \rfloor } R_{\cal I} + \sum_{{\cal I}\in \lfloor {\cal A}_1 \rfloor \cap \lfloor
{\cal A}_2 \rfloor \cap \lfloor {\cal A}_3 \rfloor} R_{\cal I} \Bigg)
\end{align*}

\noindent where the third step follows from lemma \ref{lem:HABC_expansion} in the appendix and the fourth from the
requirement $P_e^{(n)}\rightarrow 0$ (Fano's inequality) and lemma \ref{lem:IgreaterthanH} in the appendix.
\end{proof}

\noindent Applying lemma 5.1 to the sets of indices in table 1 establishes equation (\ref{eqn:BC_region}) for columns
$i=1,2,3,4,5,6,7,8,9,10,12,13,14,15$ of ${\bf G}_{BC,3}$. Unfortunately for column $i=11$ the condition that either
${\cal A}_1 \subseteq {\cal A}_2 \cup {\cal A}_3$, ${\cal A}_2 \subseteq {\cal A}_1 \cup {\cal A}_3$ or ${\cal A}_3
\subseteq {\cal A}_1 \cup {\cal A}_2$ must hold, is violated. Consequently the 11th converse bound is established in a
different fashion.

Let ${\cal A}_1,{\cal A}_2,{\cal A}_3$ be defined by the 11th row of table 1. Then

\begin{align*}
&n\left( \sum_{{\cal I}\in {\cal A}_1} R_{\cal I}^* + \sum_{{\cal I}\in {\cal A}_2} R_{\cal I}^* + \sum_{{\cal I}\in
{\cal A}_3} R_{\cal I}^* \right) \\
&\quad \ge \sum_{{\cal I}\in {\cal A}_1} H({\bf X}_{\cal I}) + \sum_{{\cal I}\in {\cal A}_2} H({\bf X}_{\cal I}) +
\sum_{{\cal
I}\in {\cal A}_3} H({\bf X}_{\cal I}) \\
&\quad \ge H\left(\cup_{{\cal I}\in {\cal A}_1} {\bf X}_{\cal I}\right) + H\left(\cup_{{\cal I}\in {\cal A}_2} {\bf
X}_{\cal I}\right) + H\left(\cup_{{\cal I}\in {\cal A}_3} {\bf X}_{\cal I}\right) \\
&\quad \ge H\left(\cup_{{\cal I}\in {\cal A}_1} {\bf X}_{\cal I}\right) + H\left(\cup_{{\cal I}\in {\cal A}_2}
{\bf X}_{\cal I}\right) + H\left(\cup_{{\cal I}\in {\cal A}_3 \cup \{123\} } {\bf X}_{\cal I}\right) - H\left( {\bf X}_{123} \right) \\
&\quad \ge H( \cup_{{\cal I}\in {\cal A}_1} {\bf X}_{\cal I} | {\bf W}_1,{\bf W}_{12},{\bf W}_{13},{\bf W}_{123}) + H({\bf W}_1,{\bf W}_{12},{\bf W}_{13},{\bf W}_{123}) \\
&\quad\quad + H(\cup_{{\cal I}\in {\cal A}_2} {\bf X}_{\cal I} | {\bf W}_2,{\bf W}_{12},{\bf W}_{23},{\bf W}_{123}) + H({\bf W}_2,{\bf W}_{12},{\bf W}_{23},{\bf W}_{123}) \\
&\quad\quad + H(\cup_{{\cal I}\in {\cal A}_3 \cup \{123\}} {\bf X}_{\cal I} | {\bf W}_3,{\bf W}_{13},{\bf W}_{23}) + H({\bf W}_3,{\bf W}_{13}, {\bf W}_{23}) - H({\bf X}_{123})\\
&\quad \ge H( {\bf X}_{123} | {\bf W}_1,{\bf W}_{12},{\bf W}_{13},{\bf W}_{123}) + H({\bf W}_1,{\bf W}_{12},{\bf W}_{13},{\bf W}_{123}) \\
&\quad\quad + H( {\bf X}_{123} | {\bf W}_2,{\bf W}_{12},{\bf W}_{23},{\bf W}_{123}) + H({\bf W}_2,{\bf W}_{12},{\bf W}_{23},{\bf W}_{123}) \\
&\quad\quad + H({\bf X}_{123} | {\bf W}_3,{\bf W}_{13},{\bf W}_{23}) + H({\bf W}_3,{\bf W}_{13}, {\bf W}_{23}) - H({\bf X}_{123})\\
&\quad \ge H({\bf W}_1,{\bf W}_{12},{\bf W}_{13},{\bf W}_{123}) + H({\bf W}_2,{\bf W}_{12},{\bf W}_{23},{\bf W}_{123})
+ H({\bf W}_3,{\bf W}_{13}, {\bf W}_{23}) \\
&\quad = H({\bf W}_1)+H({\bf W}_2)+H({\bf W}_3)+2H({\bf W}_{12}) +2H({\bf W}_{13})+2H({\bf W}_{23})+2H({\bf W}_{123}) \\
&\quad = n(R_1 + R_2 + R_3 + 2R_{12} + 2R_{13} + 2R_{23} + 2R_{123})
\end{align*}

\noindent where the fourth step follows from the requirement $P_e^(n)\rightarrow 0$ (Fano's inequality), the sixth from
lemma \ref{lem:zeroint condsets} and the seventh from the independence of the messages.

\begin{figure*}
\footnotesize \centering
\begin{tabular}{|c|c|c|c|}
\hline
$i$ & ${\cal A}_1$ & ${\cal A}_2$ & ${\cal A}_3$ \\
\hline\hline
1 & $ \{1\}, \{12\}, \{13\}, \{123\} $ & $\phi$ & $\phi$ \\
\hline
2 & $ \{2\}, \{12\}, \{23\}, \{123\}$ & $\phi$ & $\phi$ \\
\hline
3 & $ \{3\}, \{13\}, \{23\}, \{123\}$ & $\phi$ & $\phi$ \\
\hline
4 & $ \{1\}, \{2\}, \{12\}, \{13\}, \{23\}, \{123\}$ & $\phi$ & $\phi$ \\
\hline
5 & $ \{1\}, \{3\}, \{12\}, \{13\}, \{23\}, \{123\}$ & $\phi$ & $\phi$ \\
\hline
6 & $ \{2\}, \{3\}, \{12\}, \{13\}, \{23\}, \{123\}$ & $\phi$ & $\phi$ \\
\hline
7 & $ \{1\}, \{2\}, \{3\}, \{12\}, \{13\}, \{23\}, \{123\}$ & $\phi$ & $\phi$ \\
\hline
8 & $ \{1\}, \{3\}, \{12\}, \{13\}, \{123\}$ & $\{2\}, \{12\}, \{23\}, \{123\}$ & $\phi$ \\
\hline
9 & $ \{1\}, \{2\}, \{12\}, \{13\}, \{123\}$ & $\{3\}, \{13\}, \{23\}, \{123\}$ & $\phi$ \\
\hline
10 & $ \{1\}, \{2\}, \{12\}, \{23\}, \{123\}$ & $\{3\}, \{13\}, \{23\}, \{123\}$ & $\phi$ \\
\hline
11 & $ \{1\}, \{12\}, \{13\}, \{123\}$ & $\{2\}, \{12\}, \{23\}, \{123\}$ & $\{3\}, \{13\}, \{23\}$ \\
\hline
12 & $ \{1\}, \{2\}, \{12\}, \{13\}, \{123\}$ & $\{2\}, \{3\}, \{12\}, \{23\}, \{123\}$ & $\{3\}, \{13\}, \{23\}, \{123\}$ \\
\hline
13 & $ \{1\}, \{3\}, \{13\}, \{23\}, \{123\}$ & $\{1\}, \{2\}, \{12\}, \{23\}, \{123\}$ & $\{1\}, \{12\}, \{13\}, \{123\}$ \\
\hline
14 & $ \{1\}, \{2\}, \{12\}, \{13\}, \{123\}$ & $\{2\}, \{12\}, \{23\}, \{123\}$ & $\{1\}, \{3\}, \{13\}, \{23\}, \{123\}$ \\
\hline
15 & $ \{1\}, \{2\}, \{12\}, \{13\}, \{123\}$ & $\{2\}, \{3\}, \{12\}, \{23\}, \{123\}$ & $\{1\}, \{3\}, \{13\}, \{23\}, \{123\}$ \\
\hline
\end{tabular}
\caption{The $(1,1,1)$-multicast region for the broadcast channel, $L=2$.}
\end{figure*}

\section{Proof of Theorem 4.2}
The direct part of this proof is entirely analogous to the direct part for the broadcast channel. This establishes the
universal achievability of the ${\bf R}^*$-multicast region. The converse part is different. For each ${\bf R}^*$ we
present a sequence of channels. The limiting intersection of the capacity regions of these channels is the region in
equation (\ref{eqn:MAC_region}). The capacity regions of these channels are not precisely computed, but only outer
bounded in a manner sufficient to establish their limiting intersection.

\subsection{Direct part}
As this part of the proof is trivial and entirely analogous section 5.1 we only provide a sketch. In essence we need to
establish that each of the columns of ${\bf H}_{MAC,3}$ are achievable in the sense of section 5.1. The first column is
achieved by transmitting additional $M_{12}$ bits on the $W_1$ channel, the second column is achieved by transmitting
additional $M_{13}$ bits on the $W_1$ channel, the third column is achieved by transmitting additional $M_{12}$ bits on
the $W_2$ channel, and so on. The last column is achieved by lowering the rate of the $M_{123}$ message.

\subsection{Converse part}
For each ${\bf R}^*$ we present a sequence of deterministic channels with capacity region tending to the region in
equation (\ref{eqn:MAC_region}). The capacity regions of these channels are not explicitly computed, only outer
bounded, but we show the limiting outer bound is tight. The sequence is parameterized by the integer $k$.

Let ${\bf R}^*$ be given and assume its elements are rational. Denote their numerators and denominators by $N_{\cal I}$
and $D_{\cal I}$, for ${\cal I} \subseteq \{1,2,3\}$ so that ${\bf R}^* = (N_1/D_1,\dots,N_{123}/D_{123})$. Let $l =
LCM(D_1,\dots,D_{123})$. The $k$th channel is defined as follows. See figure \ref{fig:coor_channel} for a pictorial
representation. Every $k\times l$ time steps the channel takes in a triple of inputs and outputs one symbol. The input
alphabet is ${\cal X}={\cal X}_1\times {\cal X}_2 \times {\cal X}_3$ where
\begin{align*}
{\cal X}_1 &= \{0,1\}^{kN_1}\times \{0,1\}^{kN_{12}} \times \{0,1\}^{kN_{13}} \times \{0,1\}^{kN_{123}} \\
{\cal X}_2 &= \{0,1\}^{kN_2}\times \{0,1\}^{kN_{12}} \times \{0,1\}^{kN_{13}} \times \{0,1\}^{kN_{123}} \\
{\cal X}_3 &= \{0,1\}^{kN_3}\times \{0,1\}^{kN_{13}} \times \{0,1\}^{kN_{23}} \times \{0,1\}^{kN_{123}}
\end{align*}

\noindent The output alphabet is
\begin{multline*}
{\cal Y} = \{0,1\}^{kN_1} \times \{0,1\}^{kN_2} \times \{0,1\}^{kN_3} \\
\times \left( \{0,1\}^{kN_{12}} \cup \{e\} \right) \times \left( \{0,1\}^{kN_{13}} \cup \{e\} \right) \times \left(
\{0,1\}^{kN_{23}} \cup \{e\} \right) \times \left( \{0,1\}^{kN_{123}} \cup \{e\} \right)
\end{multline*}

\noindent where $e$ is an output symbol that can be thought of as an erasure. The channel thus decomposes into one with
$4\times3=12$ inputs and 7 outputs. The outputs at time $i$ are related deterministically to the inputs at time $i$ via
\begin{align*}
Y_1(i) &= X_{1,1}(i) \\
Y_2(i) &= X_{2,1}(i) \\
Y_3(i) &= X_{3,1}(i)
\end{align*}
\begin{align*}
Y_{12}(i) &= \left\{
         \begin{array}{ll}
           X_{1,12}(i) & \hbox{if $X_{1,12}(i)=X_{2,12}(i)$} \\
           e & \hbox{otherwise}
         \end{array}
       \right. \\
Y_{13}(i) &= \left\{
         \begin{array}{ll}
           X_{1,13}(i) & \hbox{if $X_{1,13}(i)=X_{3,13}(i)$} \\
           e & \hbox{otherwise}
         \end{array}
       \right. \\
Y_{23}(i) &= \left\{
         \begin{array}{ll}
           X_{2,23}(i) & \hbox{if $X_{2,23}(i)=X_{3,23}(i)$} \\
           e & \hbox{otherwise}
         \end{array}
       \right.
\end{align*}
\begin{align*}
Y_{123}(i) &= \left\{
         \begin{array}{ll}
           X_{1,123}(i) & \hbox{if $X_{1,123}(i)=X_{2,123}(i)=X_{3,123}$} \\
           e & \hbox{otherwise}
         \end{array}
       \right.
\end{align*}

\noindent The input streams thus consist of blocks of $kN_i$ bits. The output streams $Y_1(i),Y_2(i),Y_3(i)$ match
their associated input streams. The output stream $Y_{12}(i)$ matches its associated input streams if and only if the
input streams match at each bit, otherwise the erasure symbol is outputted. Likewise for the other output streams. For
this reason the boxes inside the channel in figure \ref{fig:channel_MAC} are labeled 'coordination channel'. See figure
\ref{fig:coor_channel} for a pictorial example of one such coordination channel. The idea of the coordination channels
is that in the limit of large $k$, they only let common information through. This should be intuitive from their
definition and from the figure.

\begin{figure}
\centering
\includegraphics[width=450pt]{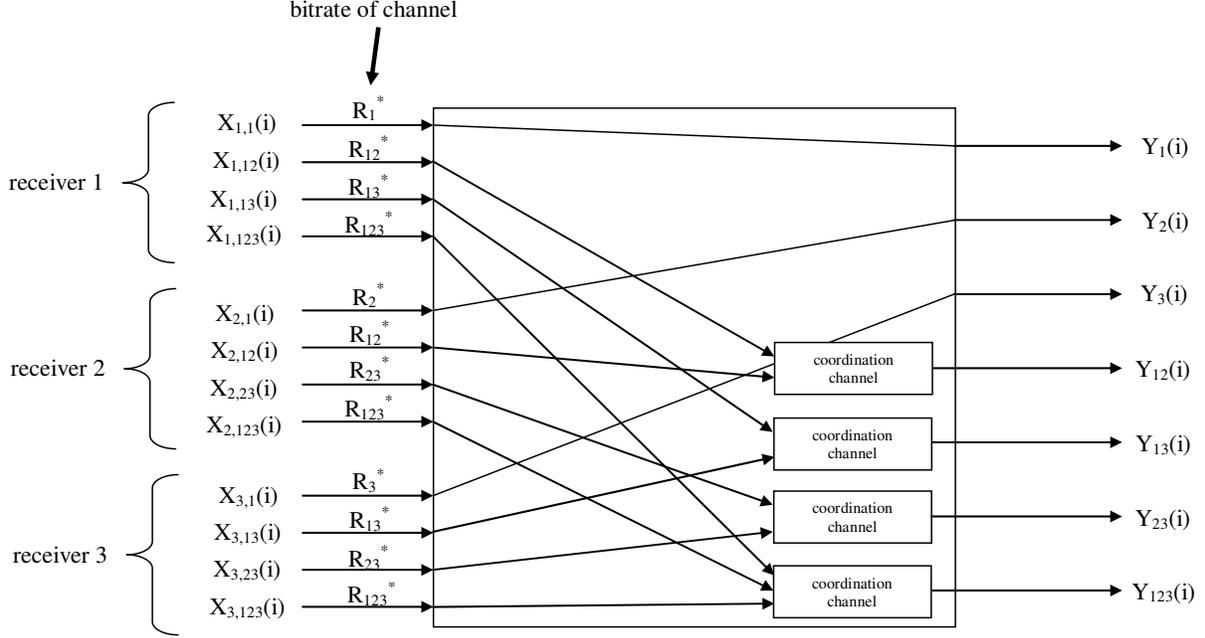}
\caption{Illustration of the deterministic multiple access channel used in converse.} \label{fig:channel_MAC}
\end{figure}

\begin{figure}
\centering
\includegraphics[width=350pt]{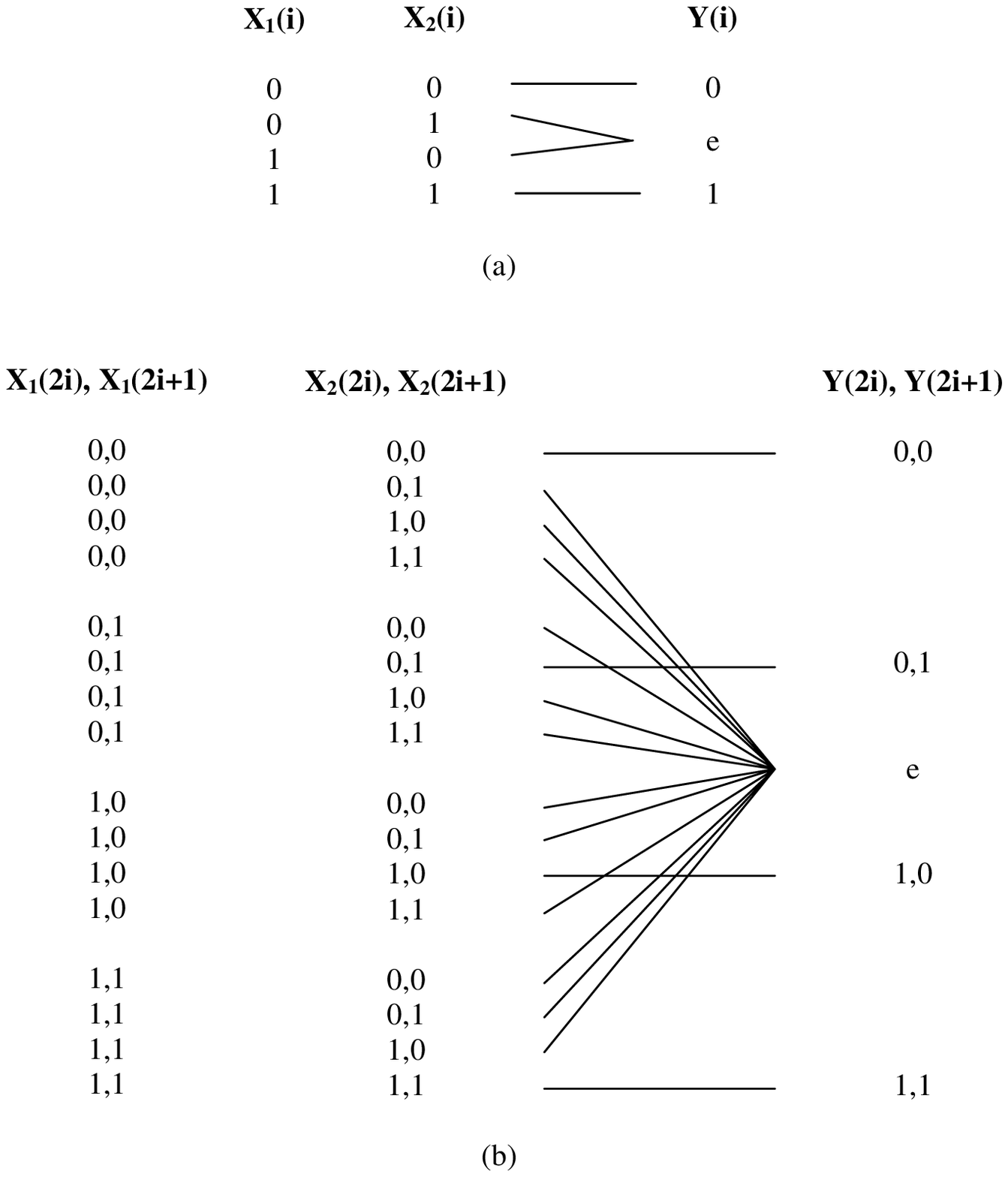}
\caption{(a) a coordination channel for $k=1$. (b) a coordination channel for $k=2$.} \label{fig:coor_channel}
\end{figure}

\noindent We now bound the capacity region of this channel. It is clear that we can further decompose the channel into
seven parallel channels, one linking $X_{1,1}$ and $Y_1$, one linking $X_{2,1}$ and $Y_2$, one linking $X_{3,1}$ and
$Y_3$, one linking $(X_{1,12},X_{2,12})$ and $Y_{12}$, one linking $(X_{1,13},X_{3,13})$ and $Y_{13}$, one linking
$(X_{2,23},X_{3,23})$ and $Y_{23}$, and one linking $(X_{1,123},X_{2,123},X_{3,123})$ and $Y_{123}$. The capacity
region of the channel in question is thus the Minkowski sum of the capacity regions of these seven channels. Denote
these seven capacity regions by ${\cal C}_{\cal I}^k$ for ${\cal I} \subseteq \{1,2,3\}$. Then the capacity region of
our channel is given by
\begin{equation*}
{\cal C}^k = \sum_{{\cal I} \subseteq \{1,2,3\}} {\cal C}_{\cal I}^k.
\end{equation*}

\noindent where sigma denotes the Minkowski sum. In particular we wish to compute the limiting intersection of these
regions
\begin{align*}
{\cal C} &= \lim_{K\rightarrow \infty} \bigcap_{k=1}^{K} {\cal C}^k \\
&= \sum_{{\cal I} \subseteq \{1,2,3\}} \lim_{K\rightarrow \infty} \bigcap_{k=1}^{K} {\cal C}_{\cal I}^k \\
&= \sum_{{\cal I} \subseteq \{1,2,3\}} {\cal C}_{\cal I}.
\end{align*}

\begin{lemma} \
\begin{enumerate}
  \item The region ${\cal C}_1$ is the set of all ${\bf R} \in {\mathbb R}^7_+$ satisfying $R_1 + R_{12} + R_{13} + R_{123} \le
R_1^*$ and $R_{\cal I} = 0$ for ${\cal I} \in \{2,3,23\}$,
  \item The region ${\cal C}_2$ is the set of all ${\bf R} \in {\mathbb R}^7_+$ satisfying $R_2 + R_{12} + R_{23} +
R_{123} \le R_2^*$ and $R_{\cal I} = 0$ for ${\cal I} \in \{1,3,13\}$,
  \item The region ${\cal C}_3$ is the set of all ${\bf R} \in {\mathbb R}^7_+$ satisfying $R_3 + R_{13} + R_{23} +
R_{123} \le R_3^*$ and $R_{\cal I} = 0$ for ${\cal I} \in \{1,2,12\}$,
  \item The region ${\cal C}_{12}$ is the set of all ${\bf R} \in {\mathbb R}^7_+$ satisfying $R_{12} + R_{123} \le
R_{12}^*$ and $R_{\cal I} = 0$ for ${\cal I} \in \{1,2,3,13,23\}$,
  \item The region ${\cal C}_{13}$ is the set of all ${\bf R} \in {\mathbb R}^7_+$ satisfying $R_{13} + R_{123} \le
R_{13}^*$ and $R_{\cal I} = 0$ for ${\cal I} \in \{1,2,3,12,23\}$,
  \item The region ${\cal C}_{23}$ is the set of all ${\bf R} \in {\mathbb R}^7_+$ satisfying $R_{23} + R_{123} \le
R_{23}^*$ and $R_{\cal I} = 0$ for ${\cal I} \in \{1,2,3,12,13\}$,
  \item The region ${\cal C}_{123}$ is the set of all ${\bf R} \in {\mathbb R}^7_+$ satisfying $R_{123} \le
R_{123}^*$ and $R_{\cal I} = 0$ for ${\cal I} \in \{1,2,3,12,13,23\}$.
\end{enumerate}
\end{lemma}

\begin{proof}
The first three regions are trivial. We establish the fourth. The messages ${\bf W}_{\cal I}$ are uniformly distributed
on $\{1,\dots,2^{nR_{\cal I}}\}$ and mutually independent for fixed $n$. Denote the $n$-length output sequence by ${\bf
Y}_{12}$. By Fano's inequality we must have $H(\cup_{{\cal I} \subseteq \{1,2,3\}} {\bf W}_{\cal I}|{\bf Y}_{12}) \le
\epsilon_n$ with $\epsilon_n\rightarrow 0$ as $n\rightarrow \infty$ in order for the error probability to be made
arbitrarily small. Thus by the mutual independence of the messages we have
\begin{align*}
n(R_1+R_2+R_3+R_{13}+R_{23}) &= H({\bf W}_1)+H({\bf W}_2)+H({\bf W}_{13})+H({\bf W}_{23}) \\
&\le H(\cup_{{\cal I} \subseteq \{1,2,3\}} {\bf W}_{\cal I},{\bf Y}_{12}) - H({\bf W}_{12},{\bf W}_{123}) \\
&\le H({\bf Y}_{12}) - H({\bf W}_{12},{\bf W}_{123}) + \epsilon_n \\
&= H({\bf Y}_{12}|{\bf W}_{12},{\bf W}_{123}) + \epsilon_n
\end{align*}

\noindent Assume for simplicity that $n=mk$ where $m$ is an integer. We write ${\bf Y}_{12} = [{\bf Y}_{12}^1,\dots,
{\bf Y}_{12}^m]$ where ${\bf Y}_{12}^i$ represents the $i$th block of $k$ symbols in $\bf Y$. Similarly ${\bf X}^i$
represents the $i$th block of $k$ symbols in $\bf X$. We also use the shorthand ${\bf W} \equiv \{{\bf W}_{12},{\bf
W}_{123}\}$. We proceed to show that $H({\bf Y}_{12}|{\bf W})$ is sufficiently small.
\begin{align*}
H({\bf Y}_{12}|{\bf W}) &\le \sum_{i=1}^m H({\bf Y}_{12}^i|{\bf W}) \\
&=-\sum_{i=1}^m \sum_{{\bf w}} P({\bf W} = {\bf w}) \sum_x P({\bf Y}_{12}^i = x|{\bf W} = {\bf w})\log P({\bf Y}_{12}^i
= x|{\bf W} = {\bf w})
\end{align*}

\noindent From the channel definition we have
\begin{equation*}
P({\bf Y}_{12}^i = x|{\bf W} = {\bf w}) = \left\{
                                            \begin{array}{ll}
                                              P({\bf X}_{1,12}^i = x, {\bf X}_{2,12}^i = x|{\bf W}={\bf w}) & \hbox{$x \neq e$;} \\
                                              P({\bf X}_{1,12}^i \neq {\bf X}_{2,12}^i|{\bf W}={\bf w}) & \hbox{$x = e$.}
                                            \end{array}
                                          \right.
\end{equation*}

\noindent Using this expression and the conditional independence of ${\bf X}_{1,12}^i$ and ${\bf X}_{2,12}^i$ given
${\bf W}$ we have

\begin{align*}
&-\sum_x P({\bf Y}_{12}^i = x|{\bf W} = {\bf w})\log P({\bf Y}_{12}^i = x|{\bf W} = {\bf w}) \\
&\quad\quad =-P({\bf X}_{1,12}^i \neq {\bf X}_{2,12}^i|{\bf W}={\bf w})\log P({\bf X}_{1,12}^i \neq {\bf X}_{2,12}^i|{\bf W}={\bf w}) \\
&\quad\quad\quad\quad - \sum_x P({\bf X}_{1,12}^i = x|{\bf W}={\bf w})P({\bf X}_{2,12}^i = x|{\bf W}={\bf w})\log P({\bf X}_{1,12}^i = x|{\bf W}={\bf w}) \\
&\quad\quad\quad\quad - \sum_x P({\bf X}_{1,12}^i = x|{\bf W}={\bf w})P({\bf X}_{2,12}^i = x|{\bf W}={\bf w})\log
P({\bf X}_{2,12}^i = x|{\bf W}={\bf w}).
\end{align*}

\noindent The first term can be upper bounded by 1 (as $-x\log_2 x < 1$ for all $x\in {\mathbb R}$). The second term
can also be upper bounded by 1. To see this, maximize first over the distribution $P({\bf X}_{2,12}^i|{\bf W}={\bf w})$
and then over the distribution $P({\bf X}_{1,12}^i|{\bf W}={\bf w})$,
\begin{align*}
&\max_{\tiny \begin{array}{c}
        P({\bf X}_{1,12}^i|{\bf W}={\bf w}) \\
        P({\bf X}_{2,12}^i|{\bf W}={\bf w})
      \end{array}}
\sum_x P({\bf X}_{1,12}^i = x|{\bf W}={\bf w})P({\bf X}_{2,12}^i = x|{\bf W}={\bf w})\log P({\bf X}_{1,12}^i = x|{\bf
W}={\bf w}) \\
&\quad\quad = \max_{P({\bf X}_{1,12}^i|{\bf W}={\bf w})} \left[ \max_{x}P({\bf X}_{1,12}^i = x|{\bf W}={\bf w}) \right]
\log \left[ \max_{x}P({\bf X}_{1,12}^i = x|{\bf W}={\bf w}) \right] \\
&\quad\quad \le 1
\end{align*}

\noindent Likewise the third term can be upper bounded by 1. Thus putting this all together we have
\begin{align*}
H({\bf Y}_{12}|{\bf W}) &< 3\sum_{i=1}^m \sum_{{\bf w}} P({\bf W} = {\bf w}) \\
&= 3m
\end{align*}

\noindent and so
\begin{align*}
R_1+R_2+R_3+R_{13}+R_{23} &< 3m/n \\
&= 3/k
\end{align*}

\noindent Then by letting $k\rightarrow \infty$ we have $R_{\cal I}=0$ for ${\cal I} \in \{1,2,3,13,23\}$. From the
structure of the coordination channel it is clear that we can achieve points $(R_{12},R_{123}) = (R_{12}^*,0)$ and
$(R_{12},R_{123}) = (0,R_{12}^*)$. By time-sharing we can achieve all points in the region $R_{12}+R_{123} \le
R_{12}^*$. Conversely from Fano's inequality we have
\begin{align*}
n(R_{12}+R_{123}) &=
H({\bf W}_{12})+H({\bf W}_{123}) \\
&\le H({\bf Y}) \\
&\le \log(2^{kR_{12}^*} + 1)^m \\
&= n(R_{12}^* + \delta_k)
\end{align*}

\noindent where $\delta_k \rightarrow 0$ as $k\rightarrow \infty$. This establishes the fourth component of the lemma.
The remaining components are established in the same manner. We omit the details.
\end{proof}

\noindent It remains to show that the region $\sum_{{\cal I}\subseteq \{1,2,3\}} {\cal C}_{\cal I}$, corresponds to the
region in equation (\ref{eqn:MAC_region}).

\section{Appendix}
\begin{lemma}\label{lem:IgreaterthanH}
Let ${\cal X}_1, {\cal X}_2$ and ${\cal X}_3$ be three sets of random variables satisfying at least one of the
properties ${\cal X}_1 \subseteq {\cal X}_2\cup {\cal X}_3$, ${\cal X}_2 \subseteq {\cal X}_1\cup {\cal X}_3$ or ${\cal
X}_3 \subseteq {\cal X}_1\cup {\cal X}_2$. Let $W$ be a random variable that satisfies $H(W|{\cal X}_i)=0$ for
$i=1,2,3$. Then
\begin{equation*}
I({\cal X}_1;{\cal X}_2;{\cal X}_3) \ge H(W)
\end{equation*}
\end{lemma}

\begin{proof}
\begin{align*}
I({\cal X}_1;{\cal X}_2;{\cal X}_3) &= I(W,{\cal X}_1;W,{\cal X}_2;W,{\cal X}_3) \\
&= H(W,{\cal X}_1)+H(W,{\cal X}_2)+H(W,{\cal X}_3) \\
&\quad\quad -H(W,{\cal X}_1,{\cal X}_2)-H(W,{\cal X}_1,{\cal X}_3)-H(W,{\cal X}_2,{\cal X}_3)+H(W,{\cal X}_1,{\cal X}_2,{\cal X}_3) \\
&= H(W) + H({\cal X}_1|W)+H({\cal X}_2|W)+H({\cal X}_3|W) \\
&\quad\quad -H({\cal X}_1,{\cal X}_2|W)-H({\cal X}_1,{\cal X}_3|W)-H({\cal X}_2,{\cal X}_3|W)+H({\cal X}_1,{\cal X}_2,{\cal X}_3|W) \\
&=H(W)+I({\cal X}_1;{\cal X}_2;{\cal X}_3|W) \\
&\ge H(W)
\end{align*}

\noindent where the last step follows from lemma \ref{lem:positive_I}.
\end{proof}

\begin{lemma}\label{lem:HABC_expansion}
\begin{equation*}
H(A)+H(B)+H(C)=H(A,B,C)+I(A,B;A,C;B,C)+I(A;B;C)
\end{equation*}
\end{lemma}

\begin{proof}
ITIP
\end{proof}

\begin{lemma}\label{lem:zeroint condsets}
Let $X_1,\dots,X_n$ be a set of mutually independent r.v's. Let ${\cal X}_1,{\cal X}_2$ and ${\cal X}_3$ be three
subsets of these r.v.'s with the property ${\cal X}_1 \cap {\cal X}_2 \cap {\cal X}_3 = \phi$. Then for any r.v. $Y$
\begin{equation*}
H(Y|{\cal X}_1) + H(Y|{\cal X}_2) + H(Y|{\cal X}_3) \ge H(Y).
\end{equation*}
\end{lemma}

\begin{proof}
\begin{align*}
&H(Y|{\cal X}_1) + H(Y|{\cal X}_2) + H(Y|{\cal X}_3) \\
&\quad\quad\ge H(Y|{\cal X}_1,{\cal X}_2^c) + H(Y|{\cal X}_2,{\cal X}_3^c) + H(Y|{\cal X}_3,{\cal X}_1^c) \\
&\quad\quad= H(Y,{\cal X}_1,{\cal X}_2^c) + H(Y,{\cal X}_2,{\cal X}_3^c) + H(Y,{\cal X}_3,{\cal X}_1^c) - H({\cal X}_1,{\cal X}_2^c) - H({\cal X}_2,{\cal X}_3^c) - H({\cal X}_3,{\cal X}_1^c) \\
&\quad\quad= H(Y,{\cal X}_1,{\cal X}_2,{\cal X}_3) + I(Y,{\cal X}_1,{\cal X}_2^c;Y,{\cal X}_2,{\cal X}_3^c;Y,{\cal X}_3,{\cal X}_1^c) \\
&\quad\quad\quad\quad+ I(Y,{\cal X}_1,{\cal X}_2,{\cal X}_3;Y,{\cal X}_1,{\cal X}_2,{\cal X}_3;Y,{\cal X}_1,{\cal X}_2,{\cal X}_3) \\
&\quad\quad\quad\quad- H({\cal X}_1,{\cal X}_2^c) - H({\cal X}_2,{\cal X}_3^c) - H({\cal X}_3,{\cal X}_1^c) \\
&\quad\quad= 2H(Y,{\cal X}_1,{\cal X}_2,{\cal X}_3) + I(Y,{\cal X}_1,{\cal X}_2^c;Y,{\cal X}_2,{\cal X}_3^c;Y,{\cal X}_3,{\cal X}_1^c) \\
&\quad\quad\quad\quad- 2H({\cal X}_1\backslash {\cal X}_2 \cup {\cal X}_3) - 2H({\cal X}_2\backslash {\cal X}_1 \cup {\cal X}_3)- 2H({\cal X}_3\backslash {\cal X}_1 \cup {\cal X}_2) \\
&\quad\quad\quad\quad - 2H({\cal X}_1 \cap {\cal X}_2 \backslash {\cal X}_3) - 2H({\cal X}_1 \cap {\cal X}_3 \backslash {\cal X}_2) - 2H({\cal X}_2 \cap {\cal X}_3 \backslash {\cal X}_1) - 3H({\cal X}_1 \cap {\cal X}_2 \cap {\cal X}_3) \\
&\quad\quad= 2H(Y,{\cal X}_1,{\cal X}_2,{\cal X}_3) + I(Y,{\cal X}_1,{\cal X}_2^c;Y,{\cal X}_2,{\cal X}_3^c;Y,{\cal X}_3,{\cal X}_1^c) \\
&\quad\quad\quad\quad- 2H({\cal X}_1\backslash {\cal X}_2 \cup {\cal X}_3) - 2H({\cal X}_2\backslash {\cal X}_1 \cup {\cal X}_3)- 2H({\cal X}_3\backslash {\cal X}_1 \cup {\cal X}_2) \\
&\quad\quad\quad\quad - 2H({\cal X}_1 \cap {\cal X}_2 \backslash {\cal X}_3) - 2H({\cal X}_1 \cap {\cal X}_3 \backslash {\cal X}_2) - 2H({\cal X}_2 \cap {\cal X}_3 \backslash {\cal X}_1) - 2H({\cal X}_1 \cap {\cal X}_2 \cap {\cal X}_3) \\
&\quad\quad= 2H(Y,{\cal X}_1,{\cal X}_2,{\cal X}_3) + I(Y,{\cal X}_1,{\cal X}_2^c;Y,{\cal X}_2,{\cal X}_3^c;Y,{\cal X}_3,{\cal X}_1^c) - 2H({\cal X}_1,{\cal X}_2,{\cal X}_3) \\
&\quad\quad\ge I(Y,{\cal X}_1,{\cal X}_2^c;Y,{\cal X}_2,{\cal X}_3^c;Y,{\cal X}_3,{\cal X}_1^c) \\
&\quad\quad\ge H(Y)
\end{align*}

\noindent where the third step follows from lemma \ref{lem:HABC_expansion}, the fourth from a set expansion made
possible by the mutual independence of the underlying r.v.'s $X_1,\dots,X_n$, the fifth from the property ${\cal X}_1
\cap {\cal X}_2 \cap {\cal X}_3 = \ \phi$, the sixth by a set relationship, and the eighth by lemma
\ref{lem:IgreaterthanH}.
\end{proof}

\begin{lemma}\label{lem:positive_I}
Let ${\cal X}_1, {\cal X}_2$ and ${\cal X}_3$ be sets of random variables. If either ${\cal X}_1 \subseteq {\cal X}_2
\cup {\cal X}_3$, ${\cal X}_2 \subseteq {\cal X}_1 \cup {\cal X}_3$ or ${\cal X}_3 \subseteq {\cal X}_1 \cup {\cal
X}_2$ then for any r.v. $W$,
\begin{equation*}
I({\cal X}_1;{\cal X}_2;{\cal X}_3 | W ) \ge 0.
\end{equation*}
\end{lemma}

\begin{proof}
Assume without loss of generality that the first {\it containment} property ${\cal X}_3 \subseteq {\cal X}_1 \cup {\cal
X}_2$ holds. Then
\begin{align*}
&I({\cal X}_1,{\cal X}_2,{\cal X}_3 | W ) \\
&\quad\quad = H({\cal X}_1|W)+H({\cal X}_2|W)+H({\cal X}_3|W) \\
&\quad\quad\quad\quad-H({\cal X}_1,{\cal X}_2|W)-H({\cal X}_1,{\cal X}_3|W)-H({\cal X}_2,{\cal X}_3|W)+H({\cal X}_1,{\cal X}_2,{\cal X}_3|W) \\
&\quad\quad = H({\cal X}_1|W)+H({\cal X}_2|W)+H({\cal X}_3|W)-H({\cal X}_1,{\cal X}_2|W)-H({\cal X}_1,{\cal X}_3|W) \\
&\quad\quad = I({\cal X}_1;{\cal X}_3|W)+I({\cal X}_2;{\cal X}_3|W)-H({\cal X}_3|W) \\
&\quad\quad \ge H({\cal X}_1 \cap {\cal X}_3|W) + H({\cal X}_2 \cap {\cal X}_3|W) - H({\cal X}_1 \cap {\cal X}_3, {\cal
X}_2 \cap {\cal X}_3|W) \\
&\quad\quad =I({\cal X}_1 \cap {\cal X}_3;{\cal X}_2 \cap {\cal X}_3|W) \\
&\quad\quad \ge 0.
\end{align*}

\noindent The second step follows from the containment property ${\cal X}_1 \subseteq {\cal X}_2 \cup {\cal X}_3$. The
first term in the fourth step follows by applying lemma \ref{lem:IgreaterthanW} with $W=W$, $X = {\cal X}_1$, $Y =
{\cal X}_3$ and $Z = {\cal X}_1\cap {\cal X}_3$, the second term by applying the same lemma with $W=W$, $X = {\cal
X}_2$, $Y = {\cal X}_3$ and $Z = {\cal X}_2\cap {\cal X}_3$. The third term in the third and fourth steps are equal by
the containment property.
\end{proof}

\begin{lemma}\label{lem:IgreaterthanW}
If $H(Z|X)=0$ and $H(Z|Y)=0$ then $I(X;Y|W)\ge H(Z|W)$.
\end{lemma}

\begin{proof}
\begin{align*}
I(X;Y|W)&=I(X,Z;Y,Z|W)\\
&=H(X,Z|W)+H(Y,Z|W)-H(X,Y,Z|W)\\
&=H(Z|W)+H(X|W,Z)+H(Z|W)+H(Y|W,Z)-H(Z|W)-H(X,Y|W,Z)\\
&=H(Z|W)+I(X;Y|W,Z)\\
&\ge H(Z|W).
\end{align*}
\end{proof}

\begin{lemma}\label{lem:IgreaterthanW}
If $H(Z|X)=0$ and $H(Z|Y)=0$ then $I(X;Y|W)\ge H(Z|W)$.
\end{lemma}

\end{document}